\documentclass[smallextended,glov3,final]{svjour3}                     

\usepackage[utf8x]{inputenc}
\usepackage[T1]{fontenc}
\usepackage{latexsym}
\usepackage{amsmath}
\usepackage{amssymb}
\usepackage{hyperref}

\usepackage{tcolorbox}

\usepackage[mathcal]{eucal}
\usepackage{multicol}
\usepackage{graphicx}

\usepackage{caption,subcaption}

\usepackage[nosort]{cite}
\usepackage{colortbl}
\newcommand{\R}{\mathbb{R}}

\begin{document}
\title{\textbf{
 Novel Corona virus Disease  infection in Tunisia: Mathematical model and the impact of the quarantine strategy}}

\author{ Haifa Ben Fredj \and   Farouk Chérif}
\institute{ Haifa Ben Fredj \at
             MaPSFA, LR11ES35, ESSTHS, University of Sousse-Tunisia    \\   
                 Tel.: +21650970204\\      
              \email{haifabenfredjh@gmail.com}  
           \and
           Farouk Chérif \at
          ISSATs and Laboratory of Mathematical Physic, Specials Functions and Applications (MaPSFA), LR11ES35, ESSTHS, University of  Sousse-Tunisia \\
              Tel.: +21629307607\\
              \email{faroukcheriff@yahoo.fr}           
}

\date{Received: date / Accepted: date}

\maketitle

\begin{abstract}\item
In this paper, we propose a new model for the dynamics of COVID-19
infections. Our approach consists of seven phenotypes: the susceptible
humans, exposed humans, infectious humans, the recovered humans, the
quarantine population, the recovered-exposed and deceased population.\\
We proved first through mathematical approach the positivity, boundness and
existence of a solution to considered model.\\
We also studied  the existence of the disease
free equilibrium and corresponding stability. Hence, the actual
reproduction number was calculated .\\
We analyzed the dependence of basic reproductions $R_0$ on the confidence of
our model. Our work shows, in particular, that the disease will decrease out
if the number of reproduction was less than one. Moreover, the impact of the
quarantine strategies to reduce the spread of this disease is discussed.\\
The theoretical results are validated by some numerical simulations of the
system of the epidemic's differential equations.

\end{abstract}

\section{Introduction}
\label{sec:1}

Since the beginning of the  COVID-19 epidemic in Wuhan City on December 2019, the opinions of scientists, researchers and commen- tators contradict each other every day. On 7 January, the coronavirus disease (COVID-2019) which was named as a severe acute respiratory syndrome coronavirus 2 (SARS-CoV-2) by International Committee on Taxonomy of Viruses on 11 February, 2020, was identified as the causative virus by Chinese authorities \cite{Wo}, and has become a pandemic especially by travelers \cite{Wi}. This forced the World Health Organization (WHO) to consider the dramatic spread of the infection in March 2020 as a public health emergency of international concern.  This epidemic is characterized by its rapid spread and its symptoms do not appear quickly. In particular, accord- ing to the WHO, the incubation period is from 2 days to 14 days \cite{Wo}. In addition, there is no anti-viral treatment or vaccination officially approved for the management or prevention of this epidemic. In order to fight this outbreak, the public health decision and policy makers should decide and  follow strategic and health-care management. Since the declaration of the first case of  COVID-19 in Tunisia in early March 2020, thousands of screenings have been carried out in Tunisia. Following the first recommendations of the WHO Tunisia, as France, Italy or Spain, first opted  for rare and targeted screenings on suspicious people, their entourage and certain sources of contamination; in particular for people arriving from abroad. Next, If a person is positive, the authorities try to trace all the people with whom he has been in contact and they are called to place themselves in self-containment without necessarily being tested. They can only be tested if they themselves have symptoms of COVID-19.\\

It is clear that the low number of screenings, mainly on "suspect" cases or those presenting significant symptoms, does not give a precise idea of the number of people who could potentially be infected without knowing them. This gap between the day of infection and the day of diagnosis can have serious consequences on the spread of the epidemic. All this shows the complexity of the situation since the unknown things related to the epidemic are more important than the things we already know.\\

These are the daily question that must be answered: how many people exactly recover from COVID-19? How many people are infected from COVID-19? How many people died from COVID-19?\\

Recently, several mathematical models have been published in order to be able to study the dynamics and the evolution of this pandemic. One can refer to \cite{Fa, Ji, Ku, Pe, Ta1, Ta} and their references. It appears from clinical experiences and recent articles that knowing the data of infected people in the population would be very useful to have better models of when disease will peak and decline, and also when we can begin to let people go back to work. Also, knowing the real number of recovered people also could indicate how easily people can build immunity against the virus.\\

Motivated by the above discussion our method comprises seven phenotypes; the susceptible humans, exposed humans, infectious humans, the recovered humans, the quarantine population, the recovered-exposed and the dead population in order to improve and adapt the susceptible-infected-recovered (SIR ) model. In fact, we noticed that people previously infected and recovered from COVID-19 are generally excluded as susceptible individuals in the modeling, which would have an impact on predicting the number of cases that will occur in the near future. In our approach we chose to integrate them into the model as a recovered-exposed population and which we will note $ E_r$. In addition, at the individual level, if people can find out if they have been or slightly infected and cured, and if they show civility and respect the distancing, they can safely return to work once the  general quarantine is off.\\

Our main contributions in this paper are:\\
–	We gave a model of dynamical behavior of COVID-19 in Tunisia.\\
–	We studied the qualitative properties of our model.\\
–	We established the stability of equilibrium point via $R_0$.\\
–	We proved theoretically and by numerical simulations the effect of quarantine strategy.\\

The reminder of the paper is organized as follows. In Section 2, we present our model. Section 3 is devoted to the mathematical analysis: Boundness, positivity, and the equilibrium point. the dynamics of exposed and infected population are also studied in this section. In Section 4, via Matab, we perform numerical simulations of three types of populations. Section 5 concludes the paper with some recommendations.

\section{ Mathematical model}
\label{sec:2}

\quad The classical susceptible-infected-recovered (SIR) model in epidemiology 
\cite{He} allows the determination of critical condition of disease
development in the population irrespective of the total population size over a
short period of time. The SIR is considered from the following simplest ODE
system: 
\begin{equation}
\begin{array}{lll}
\dfrac{dS}{dt} & =-\dfrac{\lambda }{N}SI, &  \\ 

\dfrac{dI}{dt} & =\dfrac{\lambda }{N}SI-\beta I, &  \\ 

\dfrac{dR}{dt} & =\beta I. & 
\end{array}%
\end{equation}%
for the sizes of the susceptible sub-population S, infected I, and recovered 
$R$. The term $\dfrac{\lambda}{N} IS$ describes the disease transmission rate due to
the contacts between susceptible and infected individuals, $\beta I$
characterizes the rate of recovery of infected people. It is assumed here that the recovered individuals do not
return to susceptible class, that is, recovered individuals have immunity
against the disease; they cannot become infected again and cannot infect
susceptible either.\\

Day by day, we see that the specifics of this virus require more complex
models. Roughly speaking, it is necessary to question the duration of the
incubation, or the presence of people who do not have symptoms after their
infection, but who nevertheless participate in the spread of the virus and
consider compartments for these different cases. In our work, we have used
a 7- Phenotype model which takes into account people who are infected,
exposed, recovered, those susceptible, the deceased, Quarantine but also
recovered-exposed people who are not  counted by the government.

The interaction between the above sub-populations can be described by the
compartmental diagram in Figure 2.1. The parameters indicated in Figure 1
are described in Table 1. 
\begin{figure}[hptb]
\centering
\includegraphics[scale=0.6]{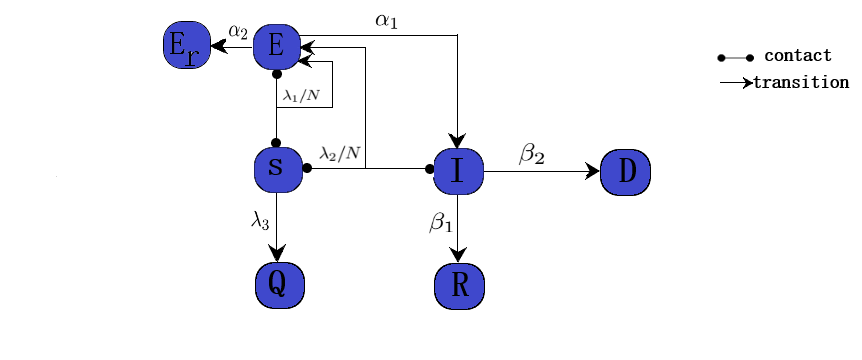}
\caption{ Model flow chart showing the compartments.}
\end{figure}

So the diagram above (figure 2.1) can be expressed as a system of non-linear differential equations (SEIRDQ) : 
\begin{equation}
\left\{ 
\begin{array}{lll}
\dfrac{dS}{dt} & =-\dfrac{\lambda_1}{N} S E-\dfrac{\lambda_2}{N} S
I-\lambda_3 S, &  \\ 
\\
\dfrac{dE}{dt} & = \dfrac{\lambda_1}{N} S E+\dfrac{\lambda_2}{N} S
I-(\alpha_1+\alpha_2) E,&  \\ 
\\
\dfrac{dI}{dt} & = \alpha_1 E-(\beta_1 +\beta_2)I, &  \\ 
\\
\dfrac{dE_r}{dt} & = \alpha_2 E, &  \\ 
\\
\dfrac{dR}{dt} & =\beta_1 I, &  \\ 
\\
\dfrac{dD}{dt} & =\beta_2 I, &  \\ 
\\
\dfrac{dQ}{dt} & =\lambda_3 S,& 
\end{array}
\right.
\end{equation}
\newpage
where
 \begin{table}[hbtp]
\centering
\begin{tabular}{|p{2cm}|p{4.6cm}|}
\hline
parameters & description \\ \hline
$\lambda_1$ & the contact rate with S and E \\ \hline
$\lambda_2$ & the contact rate with S and I \\ \hline
$\lambda_3$ & the home quarantine rate of $S$ \\ \hline
$N$ & \quad the total population \\ 
& $S+E+I+E_r+R+D+Q$ \\ \hline
$\alpha_1$ & the incubation rate \\ \hline
$\alpha_2$ & the recovered rate of $E$ \\ \hline
$\beta_1$ & the recovered rate of $I$ \\ \hline
$\beta_2$ & the death rate \\ \hline
\end{tabular}%
\caption{Parameters and their description}
\end{table}
\section{ Mathematical Analysis}

\subsection{Basic Properties of the Model}

Obviously, the system (2.2) can be written as follows 
\begin{equation}
X^{\prime}(t)=A X(t)+f(X),
\end{equation}

where 
\begin{equation*}
\begin{array}{lll}
& X(t)=%
\begin{pmatrix}
S \quad E \quad I \quad E_r \quad R \quad D \quad Q%
\end{pmatrix}%
^T , A=%
\begin{pmatrix}
-\lambda_2 & \quad 0 & \quad 0 & \quad 0 & \quad 0 & \quad 0 & \quad 0 \\ 
0 & \quad -(\alpha_1+\alpha_2) & \quad 0 & \quad 0 & \quad 0 & \quad 0 & 
\quad 0 \\ 
0 & \quad \alpha_1 & \quad -(\beta_1+\beta_2) & \quad 0 & \quad 0 & \quad 0
& \quad 0 \\ 
0 & \quad \alpha_2 & \quad 0 & \quad 0 & \quad 0 & \quad 0 & \quad0 \\ 
0 & \quad 0 & \quad 0 & \quad \beta_1 & \quad 0 & \quad 0 & \quad 0 \\ 
0 & \quad0 & \quad 0 & \quad \beta_2 & \quad 0 & \quad 0 & \quad 0 \\ 
\lambda_2 & \quad 0 & \quad0 & \quad 0 & \quad 0 & \quad 0 & \quad 0 
\end{pmatrix}%
, &  \\ 
&  &  \\ 
& f(X)=%
\begin{pmatrix}
-\dfrac{\lambda_1}{N}SE-\dfrac{\lambda_2}{N}SI\quad \dfrac{\lambda_1}{N}SE-%
\dfrac{\lambda_2}{N}SI\quad 0\quad 0\quad 0 \quad 0\quad 0%
\end{pmatrix}%
^T, & 
\end{array}%
\end{equation*}
\newline
with the initial conditions satisfying the following inequalities 
\begin{equation}
0<S(0), E(0), I(0) \text{ and }0\leq E_r, Q(0),\text{ }D(0),\text{ }R(0).
\end{equation}

\subsubsection{The boundness and  positivity of the solution}

\begin{theorem}
\item Given the non-negative initial conditions (3.2), then the solutions
S(t), E(t), I(t), $E_r(t)$,R(t), D(t) and Q(t) are non-negative for all $t\geq
0 $ and bounded.
\end{theorem}

\begin{proof}\item

First, let us  prove that the solution of the system (3.1) is positive. \\

\textbf{$\quad  1^{st}$case:}\\
Suppose that $\exists t_1>0$ such that 
$$\begin{array}{lll}
S(t_1)<0,\\
S(t)>S(t_1); \forall t<t_1,\\
I(t), E(t)>0; \forall  t \leq t_1.
\end{array}$$ 
Then we have 
$$0\geq\dfrac{dS(t_1)}{dt}=-\dfrac{\lambda_1}{N}S(t_1)E(t_1)-\dfrac{\lambda_2}{N}S(t_1)I(t_1)-\lambda_3 S(t_1)> 0,\text{ absurd}.$$
\\

\textbf{$\quad 2^{nd}$case:}\\ 
Suppose that $\exists t_1>0$ such that 
$$\begin{array}{lll}
S(t_1),E(t_1)<0,\\
E(t)>E(t_1),S(t)>S(t_1); \forall t<t_1.
\end{array}$$ 
Then we have 
$$0\geq\dfrac{dE(t_1)}{dt}+\dfrac{dS(t_1)}{dt}=-\lambda_3 S(t_1)-(\alpha_1+\alpha_2)E(t_1)>0,\text{ absurd}.$$
From the second case;  $S$ and $E$ are not both negative as the same time.\\

\textbf{$\quad 3^{rd}$case:}\\
Suppose that $\exists t_1>0$ such that 
$$\begin{array}{lll}
I(t_1)<0,\\
I(t)>I(t_1); \forall t<t_1,\\
E(t)>0;\forall t \leq t_1.
\end{array}$$ 
Then we have 
$$0\geq \dfrac{dI(t_1)}{dt}=\alpha_1 E(t_1)-(\beta_1+\beta_2)I(t_1)>0,\text{ absurd}.$$\\

\textbf{$\quad 4^{th}$case:}\\
Suppose that $\exists t_1>0$ such that 
$$\begin{array}{lll}
E(t_1)=0 ,\\
E(t)>0; \forall t<t_1,\\
S(t),I(t)>0;\forall t\leq t_1.
\end{array}$$ 
We can write E as the following form
\begin{equation}
E(t)=e^{-t(\alpha_1+\alpha_2)}E(0)+\int^t_0 e^{-(t-s)(\alpha_1+\alpha_2)}\big(\dfrac{\lambda_1}{N}S(s)E(s)+\dfrac{\lambda_2}{N}S(s)I(s)\big)ds.
\end{equation}
So, we get 
$$0=E(t_1)=e^{-t_1(\alpha_1+\alpha_2)}E(0)+\int^{t_1}_0 e^{-(t_1-s)(\alpha_1+\alpha_2)}\big(\dfrac{\lambda_1}{N}S(s)E(s)+\dfrac{\lambda_2}{N}S(s)I(s)\big)ds$$
witch is impossible.\\

\textbf{$\quad 5^{th}$case:}\\
Suppose that $\exists t_1>0$ such that 
$$\begin{array}{lll}
E(t_1)=I(t_1)=0 ,\\
E(t),I(t)>0; \forall t<t_1,\\
S(t)>0;\forall t\leq t_1.
\end{array}$$ 
We can write E as the following form
\begin{equation}
E(t)=e^{-t(\alpha_1+\alpha_2)}E(0)+\int^t_0 e^{-(t-s)(\alpha_1+\alpha_2)}\big(\dfrac{\lambda_1}{N}S(s)E(s)+\dfrac{\lambda_2}{N}S(s)I(s)\big)ds.
\end{equation}
So, we get 
$$0=E(t_1)=e^{-t_1(\alpha_1+\alpha_2)}E(0)+\int^{t_1}_0 e^{-(t_1-s)(\alpha_1+\alpha_2)}\big(\dfrac{\lambda_1}{N}S(s)E(s)+\dfrac{\lambda_2}{N}S(s)I(s)\big)ds$$
which is impossible.

Then, we get E, I and S as positive functions. Therefore, we get $E_r$, R, D and Q as increasing  functions. Thus, we have $E_r$,R, D and Q as positive functions.\\
\text{                                                                                                                                                                          }\\

Now, we prove that the solution is bounded. From the positivity of the solution, we have the following inequality
$$\begin{array}{lll}
\dfrac{dS}{dt}\leq -\lambda_3 S(t),\\
\dfrac{dS}{dt}+\dfrac{dE}{dt}\leq - m_1  (S+E); \text{ where $m_1=min\{\lambda_3,\alpha_1+\alpha_2\}$}.
\end{array}$$

Then, by integrating $$\begin{array}{lll}
S(t)\leq S(0)e^{-\lambda_3 t},\\
E(t)+S(t)\leq (S(0)+E(0))e^{-m_1 t}.
\end{array}$$\\
We can deduce that $S(t)\leq S(0)$ and $E(t)\leq S(0)+E(0)$, for all t.\\

 By the same way, we get 
 $$\begin{array}{lll}
E(t)+S(t)+I(t)\leq (S(0)+E(0)+I(0))e^{-m_2 t};\text{ where $m_2=min\{\lambda_3,\alpha_2,\beta_2+\beta_2\}$}.
 \end{array}$$
Then, $I(t)\leq S(0)+E(0)+I(0)$ for $t \leq 0$.\\

For $E_r$,$R$, $D$ and $Q$, using the recent inequalities, we have 
$$\begin{array}{lll}
\dfrac{dR}{dt}&\leq \beta_1\big(E(t)+S(t)+I(t)\big)\\
&\leq (S(0)+E(0)+I(0))e^{-m_2 t}; \text{ where $m_2=min\{\lambda_3,\alpha_2,\beta_2+\beta_2\}$},\\
\dfrac{dD}{dt}&\leq\beta_2\big( E(t)+S(t)+I(t)\big)\\
&\leq (S(0)+E(0)+I(0))e^{-m_2 t},\\
\dfrac{dQ}{dt}&\leq \lambda_3 S(0)e^{-\lambda_3t}.
\end{array} $$
From integrating
$$\begin{array}{lll}
E_r(t)&\leq E_r(0)+\dfrac{\alpha_2}{m_1}(S(0)+E(0)),\\
R(t)&\leq R(0)+\beta_1\dfrac{S(0)+E(0)+I(0)}{m_2},\\
D(t)&\leq D(0)+\beta_2 \dfrac{S(0)+E(0)+I(0)}{m_2},\\
Q(t)&\leq Q(0)+S(0).
\end{array}$$
Thus, the proof is given.
 \end{proof}

\subsubsection{Existence and uniqueness of the solution}

\begin{theorem}
(Existence and Uniqueness)

\item Under the conditions the model (3.1) possesses a unique solution.
\end{theorem}

\begin{proof}
Let us define the matrix norm  by $ |||.|||$ where for $A\in \R^7\times \R^7$:
$$ |||A|||=\rho(A) \text{  and $\rho(A)$  represents the largest eigenvalue of matrix A.} $$
Firstly, the function $X\in \R^7\mapsto  AX+f(X)$ is continuous. \\

Pose $X=(S,E,I,E_r,R,Q,D)$ and $Y=(S',E',I',E'_r,R',Q',D')$. From the bounded solution, we get
$$\begin{array}{lll}
||AX-AY+f(X)-f(Y)||_{\R^7}&\leq |||A||| \text{ }||X-Y||_{\R^7}+\dfrac{\lambda_1}{N}|S||E-E'|+\dfrac{\lambda_1}{N}|E'||S-S'|\\
&+\dfrac{\lambda_2}{N}|S||I-I'|+\dfrac{\lambda_2}{N}|I'||S-S'|\\
&\leq max\{\lambda_2,\alpha_1+\alpha_2,\beta_1+\beta_2\}||X-Y||_{\R^7}+\dfrac{\lambda_1+\lambda_2}{N} max\{|S|,|E'|,|I'|\}||X-Y||_{\R^7}\\
&\leq max\big\{\dfrac{\lambda_1+\lambda_2}{N} max\{|S|,|E'|,|I'|\},\alpha_1+\alpha_2,\beta_1+\beta_2\big\}||X-Y||_{\R^7}.
\end{array}$$
Then, the  Cauchy-Lipschitz is satisfied. And consquently, the model (3.1) possesses a unique solution \cite{Co}.
\end{proof}

\subsection{Disease free equilibrium, Reproduction number $R_0$ and extinction of
infected population}

At the disease-free state, there is no disease in the human population which
implies $E=I=R=D=Q=E_r=0$. Thus, the disease-free equilibrium of the model
(3.1) is given by 
\begin{equation*}
(S^0,E^0,I^0,E_r^0,R^0,D^0,Q^0)=(N,0,0,0,0, 0,0).
\end{equation*}

We can also verify whether or not the state-point is stable. In our case,
according the positivity of our model, we have from the first equation of
system (2.2) 
\begin{equation*}
S(t)\leq S(0) e^{-\lambda_3 t} \underset{t\rightarrow +\infty}{%
\longrightarrow }0 \quad (\neq N).
\end{equation*}
As a result, the state-point is unstable.\newline

Now, let us calculate the basic reproduction number $R_0$ which is the
average number of secondary infections caused by an infectious individual
during his or her entire period of infectiousness (Diekmann et al) \cite{Di}%
. The basic reproduction number is an important non-dimensional quantity in
epidemiology as it sets the threshold in the study of a disease both for
predicting its outbreak and for evaluating its control strategies. \newline
Using the next generation operator approach by van den Driessche and
Watmough \cite{Va}, we have

\begin{equation*}
R_0= \dfrac{\lambda_1(\beta_1+\beta_2)+\alpha_1\lambda_2}{%
(\alpha_1+\alpha_2)(\beta_1+\beta_2)}.
\end{equation*}
And we have the corresponding effective control reproduction number 
defined as in \cite{Di} 
\begin{equation}
R_0(t)= \dfrac{\lambda_1(\beta_1+\beta_2)+\alpha_1\lambda_2}{%
N(\alpha_1+\alpha_2)(\beta_1+\beta_2)} S(t).
\end{equation}
This quantity provides us with a clear index to evaluate the control strategy for any time t.

\begin{theorem}
If $R_0 < 1$ , then exposed and infected population will extinct.
\end{theorem}

\begin{proof} Let us define the Lyapunov functional
$$V(t) = \lambda_2 I(t) +(\beta_1+\beta_2) E(t).$$
The derivative of V is given by 
$$\begin{array}{lll}
\overset{\cdot}{V}(t)&=\lambda_2 \overset{\cdot}{I}(t) +(\beta_1+\beta_2) \overset{\cdot}{E}(t),\\
&=\bigg(\alpha_1\lambda_2+\lambda_1(\beta_1+\beta_2)-(\alpha_1+\alpha_2)(\beta_1+\beta_2)\bigg)E(t),\\
&\leq \bigg(\alpha_1\lambda_2+\lambda_1(\beta_1+\beta_2)-(\alpha_1+\alpha_2)(\beta_1+\beta_2)\bigg)V(t).
\end{array}$$
Thus, 
$$V(t)\leq V(0) exp\bigg(\bigg(\alpha_1\lambda_2+\lambda_1(\beta_1+\beta_2)-(\alpha_1+\alpha_2)(\beta_1+\beta_2)\bigg)t\bigg),$$
where $V(0)= \lambda_2 I(0) +(\beta_1+\beta_2) E(0).$
Therefore, for $R_0 < 1$ we have $ \alpha_1\lambda_2+\lambda_1(\beta_1+\beta_2)-(\alpha_1+\alpha_2)(\beta_1+\beta_2)< 0$. Then
$$V (t)\underset{ t\rightarrow +\infty }{\longrightarrow}0 \text{ i.e } \underset{ t\rightarrow +\infty }{lim}I(t)=0=\underset{ t\rightarrow +\infty }{lim}E(t).$$
 This shows that the disease will be extinct if $R_0 < 1$.

\end{proof}

\subsection{The effect of quarantine strategies}

\quad As  most countries of the world, Tunisia fights vigorously against the spread of
the COVID-19 epidemic with some disabilities such as the lack of medical
equipment. Therefore, in order to brake the spread of this emerging epidemic the
general quarantine is adopted. In addition, since the number of beds in
intensive care in the hospital is reduced and insufficient, the infected
population is not fully hospitalized except in cases with fairly severe
symptoms. Besides, the rest of the population should respect total isolation. The
natural question here is : What is the impact of this decision in other
words this quarantine strategy?

\text{                                                                                                                                               }\\

By equation one in the system (2.2), the susceptible population can be evaluate
from the following expression \\
\begin{equation}
S(t)=e^{-\lambda _{3}t}S(0)-e^{-\lambda _{3}t}\int_{0}^{t}e^{\lambda _{3}s}%
\big[\dfrac{\lambda _{1}}{N}S(s)E(s)+\dfrac{\lambda _{2}}{N}S(s)I(s)\big]ds.
\end{equation}\\

It is clear that the parameter $\lambda _{3}$ reduces the number of people
exposed to infection with the disease. From the evolution of exponential
function, we can deduce that when the rate $\lambda_3$  is getting bigger, the susceptible population is getting smaller.

When the $\lambda _{3}=0$,  we have this diagram \\
\begin{figure}[hbtp]
\centering
\includegraphics[scale=0.6]{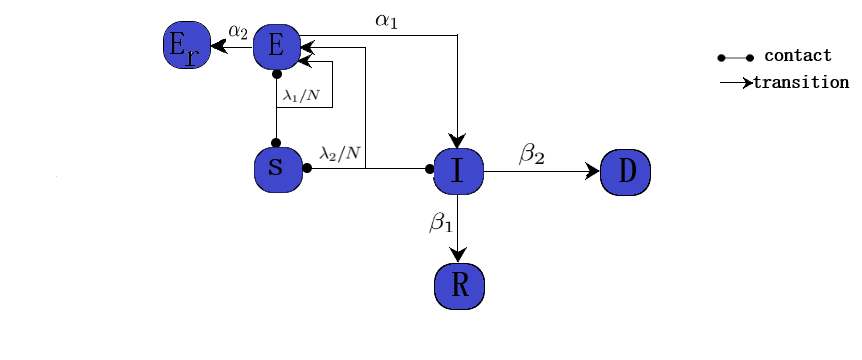}
\end{figure}\\
and the corresponding non-linear differential equations
\begin{equation}
\left\{ 
\begin{array}{lll}
\dfrac{dS}{dt} & =-\dfrac{\lambda_1}{N} S E-\dfrac{\lambda_2}{N} SI, &  \\ 
\\
\dfrac{dE}{dt} & = \dfrac{\lambda_1}{N} S E+\dfrac{\lambda_2}{N} S
I-(\alpha_1+\alpha_2) E, &  \\ 
\\
\dfrac{dI}{dt} & = \alpha_1 E-(\beta_1 +\beta_2)I, &  \\ 
\\
\dfrac{dE_r}{dt} & = \alpha_2 E, &  \\ 
\\
\dfrac{dR}{dt} & =\beta_1 I, &  \\ 
\\
\dfrac{dQ}{dt} & =\lambda_3 S, & 
\end{array}
\right.
\end{equation}
where the parameters were defined in the table 1.\\
Clearly, by the first equation, one has
\begin{equation}
S(t)=S(0)-\int_{0}^{t}\dfrac{\lambda _{1}}{N}S(s)E(s)+\dfrac{\lambda _{2}}{N}%
S(s)I(s)ds.
\end{equation}\\
Hence,\\
\begin{equation}
\int_{0}^{t}\dfrac{\lambda _{1}}{N}S(s)E(s)+\dfrac{\lambda _{2}}{N}%
S(s)I(s)ds=S(0)-S(t).
\end{equation}\\
So the quantity\\
$$
\int_{0}^{t}\dfrac{\lambda _{1}}{N}S(s)E(s)+\dfrac{\lambda _{2}}{N}%
S(s)I(s)ds=0
$$
if and only 
$$S(0)=S(t)\text{ for a certain }t,$$
which implies that  the disease almost infected every susceptible person.%
 \text{                                                                                                                                        }\\
 
This theoretical result which justifies the quarantine strategy will be
consolidated by a numerical study in the following paragraph.

\section{Numerical Results and Discussion}
\label{sec:3}
Using data of Tunisia from March $14^{th}$ to April $8^{th}$ 2020, we can estimate the parameters'values. The results can be summarized as follows:\\

\begin{table}[hbtp]
\centering
\begin{tabular}{|p{3cm}p{4cm}|}
\hline
parameters & values \\ \hline
$\lambda_1$ & 0.8 $(day^{=1})$ \\ \hline
$\lambda_2$ & 0.02 $(day^{=1})$ \\ \hline
$\lambda_3$ & 0.166 $(day^{=1})$ \\ \hline
$N$ & $11*10^6$ persons \\ \hline
$\alpha_1$ & 0.0109 $(day^{=1})$ \\ \hline
$\alpha_2$ & 0.1$(day^{=1})$ \\ \hline
$\beta_1$ & 0.003 $(day^{=1})$ \\ \hline
$\beta_2$ & 0.0037 $(day^{=1})$ \\ \hline
$R_0$ & 0.75 \\ \hline
\end{tabular}%
\caption{SEIRDQ Model parameters}
\end{table}
And for stimulation, we chose the following initial condition: S(0) =
N-E(0)-I(0), E(0) =200, I(0) =18, $E_r(0)= R(0) = D(0) = 0$.
\text{                                                                                                   }\\

In addition to the  population plots in figures 2 and 3, we collected some
meaningful quantitative information about the model parameters (table 1) and
the peak values for infected and exposed populations (table 3).\newline
The results can be summarized as follows: \\
\begin{table}[hbtp]
\centering
\begin{tabular}{|p{2.8cm}|p{3cm}|}
\hline
Peak infected & Peak exposed \\ 
day \text{ }Number(\%) & day \text{ } Number(\%) \\ \hline
38 \text{ } 648($5,9.10^{-3}$) & 12 \text{ } 3405($3,1.10^{-2}$) \\ \hline
\end{tabular}%
\caption{Infected and exposed peak values in Tunis region.}
\end{table}\\  

And we have $R_0<1$, therefore the disease will disappear.
\newpage
\begin{figure}[hbtp]
\centering
\includegraphics[scale=0.4]{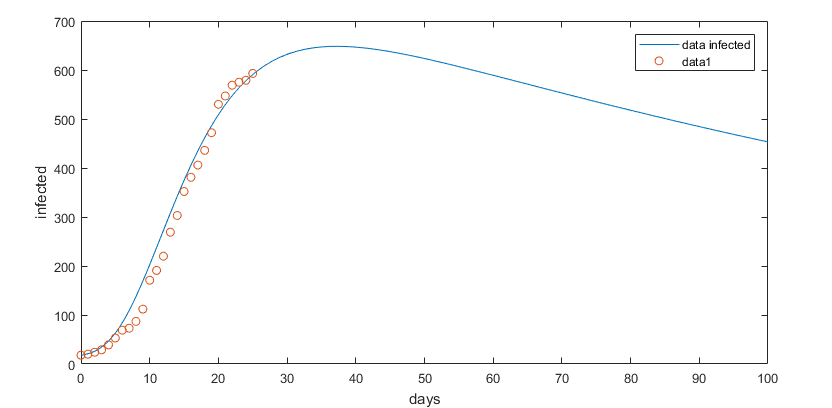}
  \caption{Infected population}
\end{figure}
\begin{figure}[hbtp]
\centering
\includegraphics[scale=0.42]{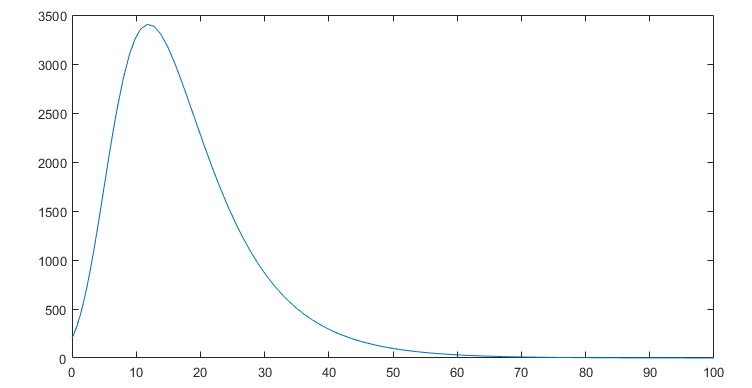}
  \caption{Exposed population}
\end{figure}
\begin{figure}[hbtp]
\centering
\includegraphics[scale=0.42]{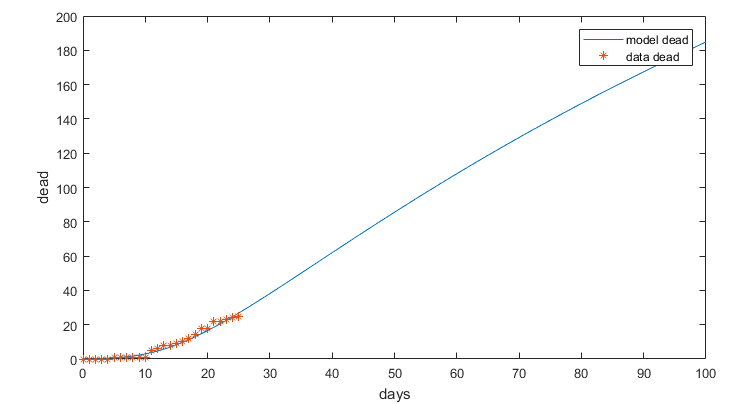}
\caption{dead population}
\end{figure}

The human body needs around 10 days to get rid of the organism. However, cases where the body needed 21 days to get rid of symptoms are detected. So after the peak, it is logical that the rate of recovered population increases. Then we can get the following figures.\\ 
\text{                                                                                                                  }\\

From figures 2 and 5, if the Tunisian people maintain this
pattern. After 50 days, the exposed population is fewer and the infected
population surpasses the peak and decreases quickly. So, the government of Tunisia can ease up the measures that it took in connection with the outbreak of the epidemic like the general quarantine.
\newpage

\begin{figure}[hbtp]
\centering
  \includegraphics[scale=0.44]{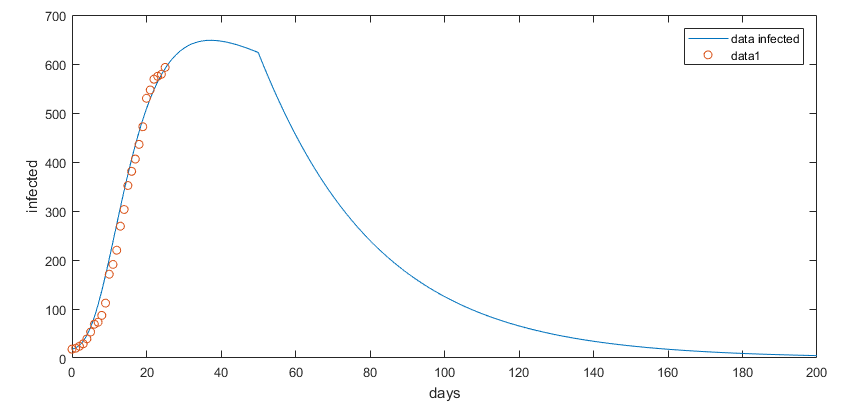}
  \caption{Infected population}
  \end{figure}
\begin{figure}[hbtp]
    \centering
   \includegraphics[scale=0.47]{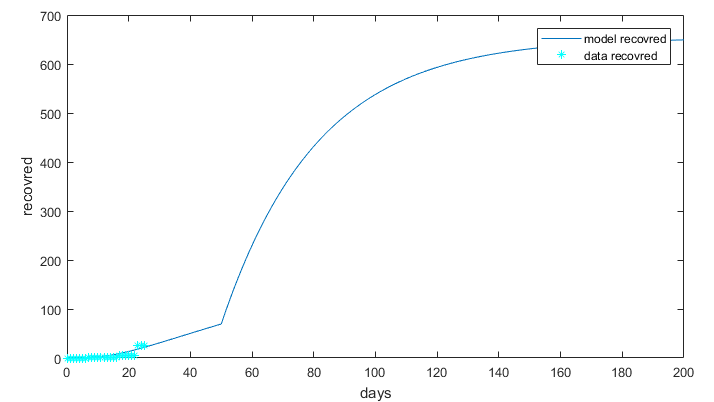}
   \caption{Recovered population}
    \end{figure}
\begin{figure}[hbtp]
    \centering
   \includegraphics[scale=0.42]{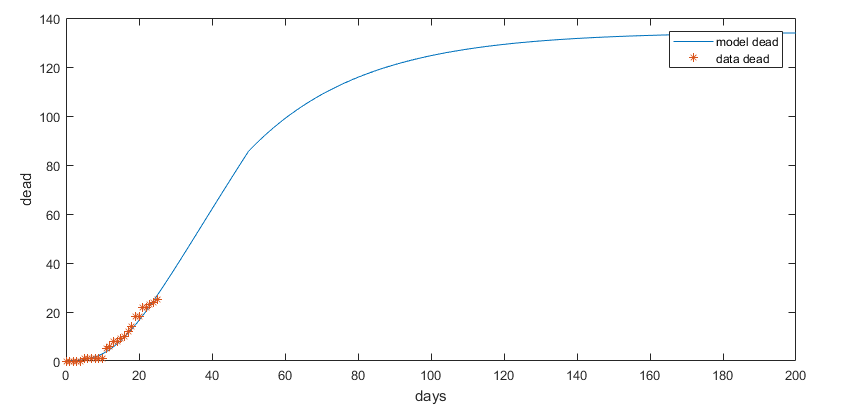}
\caption{Dead population}
\end{figure}

\text{                                                                                                  }\\
To make a better illustration of quarantine strategy, we tested different home quarantine rates ($\lambda_3$)
in figures (8.a-h).

\begin{figure}[hptb]
\begin{subfigure}[c]{9cm}
    \centering
    \includegraphics[scale=0.4]{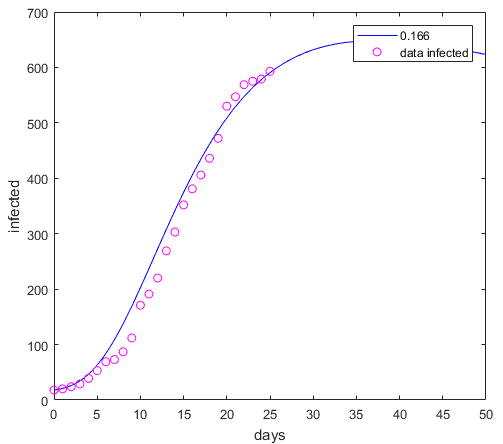}
   \caption{}
   \end{subfigure}\hfill 
\begin{subfigure}[c]{8cm}
    \centering
   \includegraphics[scale=0.4]{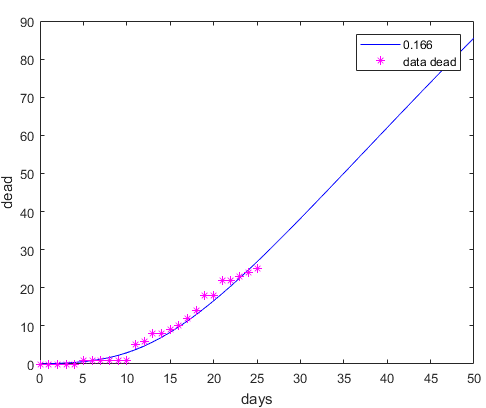}
   \caption{}
    \end{subfigure}\hfill
    
    \begin{subfigure}[c]{9cm}
    \centering
    \includegraphics[scale=0.4]{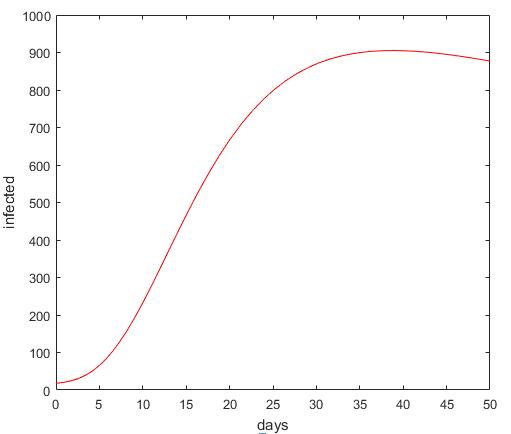}
   \caption{}
   \end{subfigure}\hfill 
\begin{subfigure}[c]{8cm}
    \centering
   \includegraphics[scale=0.4]{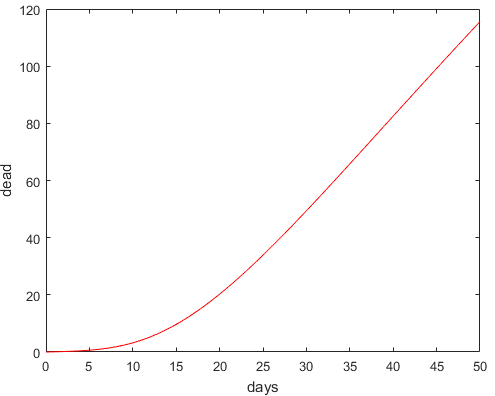}
   \caption{}
    \end{subfigure}\hfill
    \begin{subfigure}[c]{9cm}
    \centering
    \includegraphics[scale=0.4]{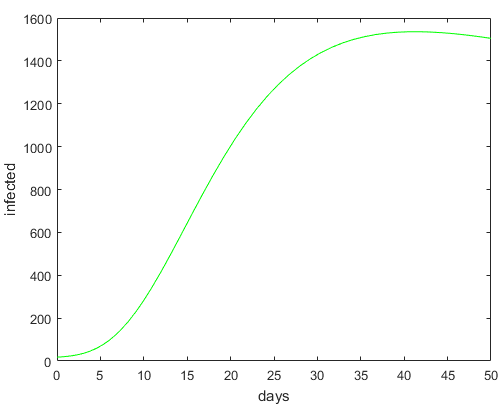}
   \caption{}
   \end{subfigure}\hfill 
\begin{subfigure}[c]{8cm}
    \centering
   \includegraphics[scale=0.4]{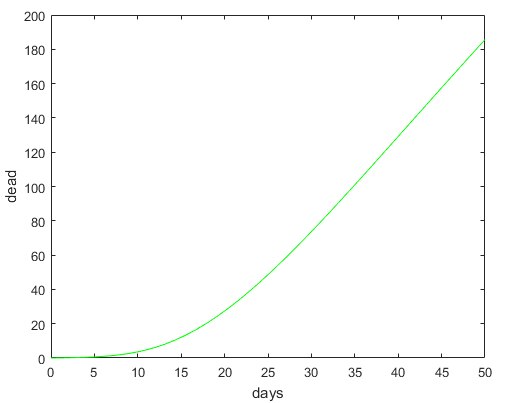}
   \caption{}
    \end{subfigure}\hfill
    
    \begin{subfigure}[c]{9cm}
    \centering
    \includegraphics[scale=0.4]{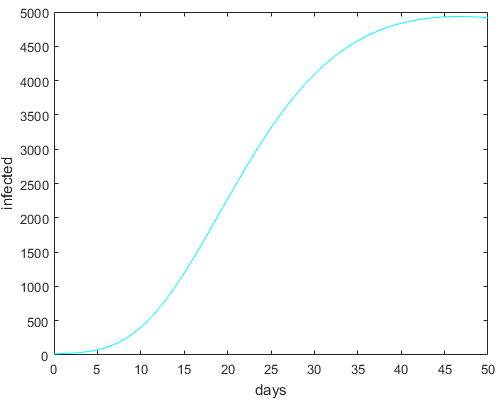}
   \caption{}
   \end{subfigure}\hfill 
\begin{subfigure}[c]{8cm}
    \centering
   \includegraphics[scale=0.4]{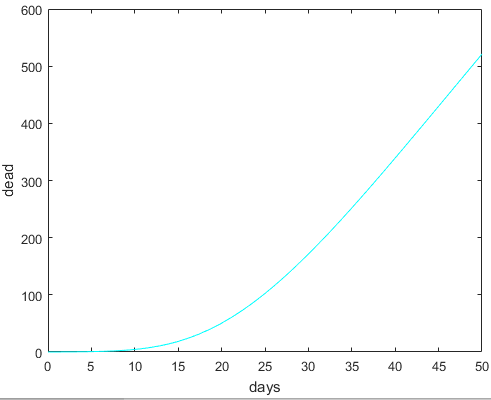}
   \caption{}
    \end{subfigure}
\caption{a-b) Infected population and  Dead population with $%
\protect\lambda_3=0.166$ which was our case, c-d) Infected population and  Dead population with $%
\protect\lambda_3=0.15$, e-f) Infected population and  Dead population with $%
\protect\lambda_3=0.13$, g-h) Infected population and  Dead population with $%
\protect\lambda_3=0.1$}
\end{figure}
\newpage

Now, let’s summarize by giving different cases in one figure in order to compare them.

\begin{figure}[hptb]
\begin{subfigure}[c]{9cm}
    \centering
    \includegraphics[scale=0.4]{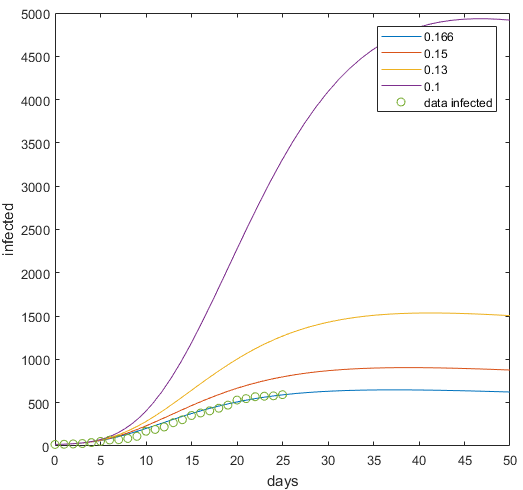}
   \caption{}
   \end{subfigure}\hfill 
\begin{subfigure}[c]{8cm}
    \centering
   \includegraphics[scale=0.35]{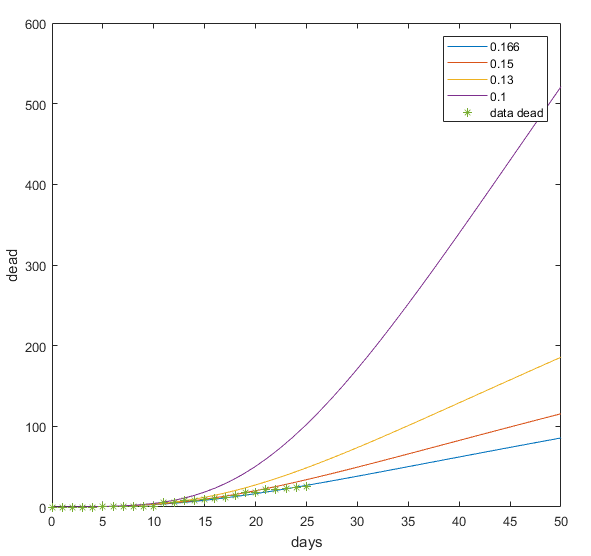}
   \caption{}
    \end{subfigure}
\caption{a) Infected population and b) Dead population with different $%
\protect\lambda_3=0.166,0.15,0.13,0.1$ }
\end{figure}

The change in the $\lambda_3$ parameter generates a consequent variation in the rate of the infected population and the deceased persons (fig.  8.a-h).

The most spectacular is the case $\lambda_3 = 0$ which corresponds to Q = 0; or in other words, in the absence of a quarantine strategy. We can see that on the 20th day of the epidemic we get $3,75$x$10^5$ infected people (fig. 9.a).
This number is far above the capacity of health structures in Tunisia.
Consequently, the choice of containment strategy has significantly reduced the number of deceased and infected persons. 
\begin{figure}[tbph]
\begin{subfigure}[c]{9cm}
    \centering
    \includegraphics[scale=0.4]{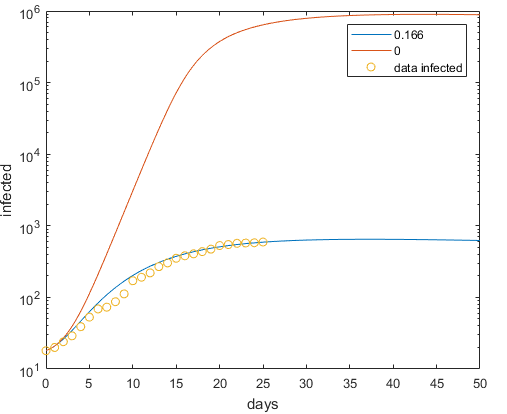}
   \caption{}
   \end{subfigure}\hfill 
\begin{subfigure}[c]{8cm}
    \centering
   \includegraphics[scale=0.4]{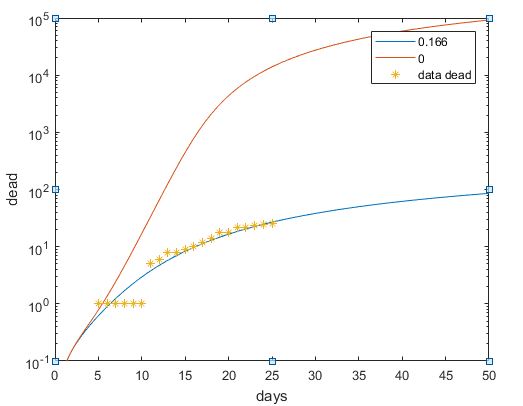}
   \caption{}
    \end{subfigure}
\caption{a) Infected population and b) Dead population with different $%
\protect\lambda _{3}=0.166,0$ }
\end{figure}

\section{Conclusion}
\label{sec:4}
\quad In this article, we presented a new model based on nonlinear differential equations allowing to model this COVID-19 in Tunisia. The aim of our work, at first, is to provide initial ideas and guidelines for a quantitative and qualitative study of our considered model. In particular, the positivity, boundness and the existence of a solution are established. The impact of the quarantine strategy is also established by two methods. We hope our work motivates new research to produce more elaborate and precise methods related to this model or to give significant improvements..

\end{document}